\documentclass{article} 
\usepackage{times}
\usepackage{hyperref}
\usepackage{url}
\usepackage{times}
\usepackage{helvet}
\usepackage{courier}
\usepackage{paralist}
\usepackage{latexsym}
\usepackage{url}
\usepackage[all]{xy}
\usepackage{amsmath}
\usepackage{amssymb}
\usepackage{comment}
\usepackage{xcolor}
\usepackage{pifont}
\usepackage{times}
\usepackage{graphicx}
\usepackage{savesym}
\savesymbol{AND}
\usepackage{xspace}
\usepackage{tikz}
\usepackage{pgfplots}
\usepackage{pgf}
\usepackage[numbers]{natbib}
\usepackage{algorithm}
\usepackage{algorithmic}
\usepackage{xspace}
\usepackage{comment}
\usepackage{footnote}
\usepackage{placeins}
\usetikzlibrary{intersections}
\usetikzlibrary{arrows,calc,fit,patterns,plotmarks,shapes.geometric,shapes.misc,shapes.symbols,   shapes.arrows,   shapes.callouts,   shapes.multipart,   shapes.gates.logic.US,   shapes.gates.logic.IEC,   er,   automata,   backgrounds,   chains,   topaths,   trees,   petri,   mindmap,   matrix,   calendar,   folding, fadings,   through,   positioning,   scopes,   decorations.fractals,   decorations.shapes,   decorations.text,   decorations.pathmorphing,   decorations.pathreplacing,   decorations.footprints,   decorations.markings, shadows,circuits}
\tikzstyle{decision}=[diamond,draw]
\tikzstyle{line}=[draw]
\tikzstyle{elli}=[draw,ellipse]
\tikzstyle{arrow} = [thick]
\newcommand{\mb}{\mbox{ }}
\newcommand{\one}{\mathbf{1}}
\newcommand{\nn}{\nonumber}

\newcommand{\R}{\mathbf{R}}

\newcommand{\ra}{\rightarrow}

\newcommand{\E}{\mathbf{E}}

\newcommand{\F}{\mathcal{F}}
\newcommand{\et}{||\Gamma\bj-\tg\bj||_\infty}

\newenvironment{proof}{{\bf Proof:} }{}
\newtheorem{theorem}{Theorem}
\newtheorem{lemma}[theorem]{Lemma}
\newtheorem{assumption}{Assumption}
\newtheorem{definition}[theorem]{Definition}

\newtheorem{corollary}{Corollary}
\newtheorem{remark}{Remark}

\newcommand{\tj}{\tilde{J}_c}
\newcommand{\hj}{\hat{J}_c}

\newcommand{\bj}{\bar{J}}
\newcommand{\tv}{\tilde{V}}
\newcommand{\hv}{\hat{V}}

\newcommand{\hu}{\hat{u}}

\newcommand{\hr}{\hat{r}_c}
\newcommand{\tr}{\tilde{r}_c}

\newcommand{\tg}{\tilde{\Gamma}}

\title{A Generalized Reduced Linear Program for Markov Decision Processes}
\author{Chandrashekar Lakshminarayanan and Shalabh Bhatnagar,\\ Department Computer Science and Automation,\\ Indian Institute of Science, Bangalore-560012, India.\\\{chandrul,shalabh\}@csa.iisc.ernet.in}

\begin{document}
\maketitle
\begin{abstract}
Markov decision processes (MDPs) with large number of states are of high practical interest. However, conventional algorithms to solve MDP are computationally infeasible in this scenario. Approximate dynamic programming (ADP) methods tackle this issue by computing approximate solutions. A widely applied ADP method is approximate linear program (ALP) which makes use of linear function approximation and offers theoretical performance guarantees. Nevertheless, the ALP is difficult to solve due to the presence of a large number of constraints and in practice, a reduced linear program (RLP) is solved instead. The RLP has a tractable number of constraints sampled from the orginial constraints of the ALP. Though the RLP is known to perform well in experiments the theoretical guarantees are available only for a specific RLP obtained under idealized assumptions.\\
In this paper, we generalize the RLP to define a generalized reduced linear program (GRLP) which has a tractable number of constraints that are obtained as positive linear combinations of the original constraints of the ALP. The main contribution of this paper is the novel theoretical framework developed to obtain error bounds for any given GRLP. Central to our framework are two $\max$-norm contraction operators. Our result solves theoretically justifies linear approximation of constraints. We discuss the implication of our results in the contexts of ADP and reinforcement learning. We also demonstrate via an example in the domain of controlled queues that the experiments conform to the theory.
\end{abstract}
\section{Introduction}\label{intro}
Markov decision processes (MDPs) is an important mathematical framework to study optimal sequential decision making problems that arise in science and engineering. Solving an MDP involves computing the optimal \emph{value-function} ($J^*$), a vector whose dimension is the number of states. MDPs with small number of states can be solved easily by conventional solution methods such as value/ policy iteration or linear programming (LP) \cite{BertB}. Dynamic programming is at the heart of all the conventional solution methods for MDPs.\\
\indent The term \emph{curse-of-dimensionality} (or in short \emph{curse}) denotes the fact that the number of states grows exponentially in the number of state variables. Most practical MDPs suffer from the curse, i.e., have large number of states and the $J^*$ is difficult to compute. A practical way to tackle the curse is to compute an approximate value function $\tilde{J}$ instead of $J^*$. The methods that compute $\tilde{J}$ instead of $J^*$ are known as approximate dynamic programming (ADP) methods whose success depends on the quality of approximation, i.e., on the quantity $||J^*-\tilde{J}||$. Most ADP methods employ linear function approximation (LFA), i.e., let $\tilde{J}=\Phi r^*$, where $\Phi$ is a feature matrix and $r^*$ is a learnt weight vector. Dimensionality reduction is achieved by choosing $\Phi$ to have fewer columns in comparison to the number of states and this makes computing $\tilde{J}$ easier.\\
\indent Approximate linear program (ALP) \cite{ALP,CS,SALP,ALP-Bor,gkp,fs,npalp} employs LFA in the linear programming formulation (\cite{BertB,NDP}) of MDP. The ALP computes an approximate value function and offers sound theoretical guarantees. A serious shortcoming of the ALP is the large number of constraints (of the order of the number of states). A technique studied in literature that tackles the issue of large number of constraints is constraint sampling \cite{CS,CST} wherein one solves a reduced linear program (RLP) with a small number of constraints sampled from the constraints of the ALP. \cite{CS} presents performance guarantees for the RLP when the constraints are sampled with respect to the stationary distribution of the optimal policy. Such an idealized assumption on the availability of the optimal policy (which in turn requires knowledge of $J^*$) is a shortcoming. Nevertheless, the RLP has been shown to perform empirically well (\cite{CS,ALP,SALP}) even when the constraints are not sampled using the stationary distribution of the optimal policy.\\
Motivated by the gap between the limited theoretical guarantees of the RLP and its successful practical efficacy, in this paper we provide a novel theoretical framework to characterize the error due to constraint reduction/approximation. The novelty and salient points of our contributions are listed below:
\begin{enumerate}
\item We define a generalized reduced linear program (GRLP) which has a tractable number of constraints that are obtained as positive linear combinations of the original constraints of the ALP.
\item We develop a novel analytical framework in order to relate $\hat{J}$, the solution to the GRLP, and the optimal value function $J^*$. In particular, we come up with two novel $\max$-norm contraction operators called the least upper bound (LUB) projection operator and the approximate least upper bound projection operator (ALUB).
\item We show that $||J^*-\hat{J}||\leq (c_1+c_2)$, where $c_1>0$, $c_2>0$ are constants. While the term $c_1$ corresponds to the error inherent to the ALP itself, the term $c_2$ constitutes the additional error introduced due to constraint approximation.
\item The results from the GRLP framework solves the problem of theoretically justifying linear approximation of constraints. Unlike the bounds in \cite{CS} that hold only for specific RLP 
our bounds hold for any GRLP and as a result any RLP.
\item We also discuss qualitatively the relative importance of our results in the context of ADP and their implication in the reinforcement learning setting.
\item We demonstrate via an example in controlled queues that the experiments conform to the theory developed.
\end{enumerate}
The rest of the paper is organized as follows. First, we present the basics of MDPs. We then discuss the ALP technique, the basic error bounds as well as, the issues and proposed solutions in literature, following by which we present the open questions we solve in this paper. We then present the main results of the paper namely the GRLP and its error analysis. We then present a qualitative discussion of our result followed by the numerical example.
\section{Markov Decision Processes (MDPs)}
In this section, we briefly discuss the basics of Markov Decision Processes (MDPs) (the reader is referred to \cite{BertB,Puter} for a detailed treatment).\\
\textbf{The MDP Model:} An MDP is a $4$-tuple $<S,A,P,g>$, where $S$ is the state space, $A$ is the action space, $P$ is the probability transition kernel and $g$ is the reward function. We consider MDPs with large but finite number of states, i.e., $S=\{1,2,\ldots,n\}$ for some large $n$, and the action set is given by $A=\{1,2,\ldots,d\}$. For simplicity, we assume that all actions are feasible in all states. The probability transition kernel $P$ specifies the probability $p_a(s,s')$ of transitioning from state $s$ to state $s'$ under the action $a$. We denote the reward obtained for performing action $a\in A$ in state $s\in S$ by $g_a(s)$.\\
\textbf{Policy:} A policy $\mu$ specifies the action selection mechanism, and is described by the sequence $\mu=\{u_1,u_2,\ldots,u_n,\ldots\}$, where $u_n\colon S \ra A, \mbox{ }\forall n \geq 0$. A stationary deterministic policy (SDP) is one where $u_n\equiv u, \mbox{ }\forall n\geq 0$ for some $u\colon S \ra A$. By abuse of notation we denote the SDP by $u$ itself instead of $\mu$. In the setting that we consider, one can find an SDP that is optimal \cite{BertB,Puter}. In this paper, we restrict our focus to the class $U$ of SDPs.  Under an SDP $u$, the MDP is a Markov chain with probability transition kernel $P_u$.\\
\textbf{Value Function:} Given an SDP $u$, the infinite horizon discounted reward corresponding to state $s$ under $u$ is denoted by $J_u(s)$ and is defined by
\begin{align}
J_u(s)\stackrel{\Delta}{=}\E[\sum_{n=0}^\infty \alpha^n g_{a_n}(s_n)|s_0=s,a_n=u(s_n)\mbox{ }\forall n\geq 0],\nn
\end{align}
where $\alpha \in (0,1)$ is a given discount factor. Here $J_u(s)$ is known as the value of the state $s$ under the SDP $u$, and the vector quantity $J_u\stackrel{\Delta}{=}(J_u(s), \forall s\in S)\in R^n$ is called the value-function corresponding to the SDP $u$.\\
\textbf{The optimal policy} $u^*$ is obtained as $u^*(s)\stackrel{\Delta}{=}\arg\max_{u\in U}J_u(s)$\footnote{Such $u^*$ exists and is well defined in the case of infinite horizon discounted reward MDP, for more details see \cite{Puter}.}.\\
\textbf{The optimal value-function} $J^*$ is the one obtained under the optimal policy, i.e., $J^*=J_{u^*}$.\\
\textbf{The Bellman Equation and Operator:} Given an MDP, our aim is to find the optimal value function $J^*$ and the optimal policy $u^*$. The optimal policy and value function obey the Bellman equation (BE) as under: $\forall s \in S$,
\begin{subequations}\label{bell}
\begin{align}
\label{bellval}J^*(s)&=\max_{ a\in A}\big(g_a(s)+\alpha \sum_{s'}p_a(s,s')J^*(s')\big),\\
\label{bellpol}u^*(s)&=\arg\max_{ a\in A}\big(g_a(s)+\alpha \sum_{s'}p_a(s,s')J^*(s')\big).
\end{align} 
\end{subequations}
Typically $J^*$ is computed first and $u^*$ is obtained by substituting $J^*$ in \eqref{bellpol}.\\
The Bellman operator $T\colon \R^n \ra \R^n$ is defined using the model parameters of the MDP as follows:
\begin{align}
(TJ)(s)=\max_{a \in A}\big(g_a(s)+\alpha \sum_{s'} p_a(s,s')J(s')\big), \mb\text{where}\mb J\in \R^n.\nn
\end{align}
\textbf{Basis Solution Methods:} When the number of states of the MDP is small, $J^*$ and $u^*$ can be computed exactly using conventional methods such as value/policy iteration and linear programming (LP) \cite{BertB}.\\
\textbf{Curse-of-Dimensionality} is a term used to denote the fact that the number of states grows exponentially in the number of state variables. Most MDPs occurring in practice suffer from the curse, i.e., have large number of states and it is difficult to compute $J^*\in \R^n$ exactly in such scenarios.\\
\textbf{Approximate Dynamic Programming} \cite{lspi,lspe,ALP,wang2014approximate}(ADP) methods compute an approximate value function $\tilde{J}$ instead of $J^*$. In order to make the computations easier ADP methods employ function approximation (FA) where in $\tilde{J}$ is chosen from a parameterized family of functions. The problem then boils down to finding the optimal parameter which is usually of lower dimension and is easily computable.\\
\textbf{Linear Function Approximation (LFA)} \cite{ALP,lspe,fourier,krylov,proto} is a widely used FA scheme such that the approximate value function $\tilde{J}=\Phi r^*$, where $\Phi=[\phi_1|\ldots|\phi_k]$ is an $n\times k$ feature matrix and $r^*$ is the parameter to be learnt.
\section{Approximate Linear Programming}
We now present the linear programming formulation of the MDP which forms the basis for ALP. The LP formulation is obtained by unfurling the $\max$ operator in the BE in \eqref{bell} into a set of linear inequalities as follows:
\small
\begin{align}\label{mdplp}
\min_{J\in \R^n} &c^\top J\nn\\
\text{s.t}\mb &J(s)\geq g_a(s)+\alpha\sum_{s'}p_a(s,s')J(s'), \mb\forall s\in S, a \in A,
\end{align}
\normalsize
where $c\in \R^n_+$ is a probability distribution and denotes the relative importance of the various states. One can show that $J^*$ is the solution to \eqref{mdplp} \cite{BertB}. The LP formulation in \eqref{mdplp} can be represented in short\footnote{$J\geq TJ$ is a shorthand for the $nd$ constraints in \eqref{mdplp}. It is also understood that constraints $(i-1)n+1,\ldots,in$ correspond to the $i^{th}$ action.} as,
\begin{align}\label{mdplpshort}
\min_{J\in \R^n} &c^\top J\nn\\
\text{s.t}\mb &J\geq T J.
\end{align}
The approximate linear program (ALP) is obtained by making use of LFA in the LP, i.e., by letting $J=\Phi r$ in \eqref{mdplpshort} and is given as
\begin{align}\label{alp}
\min_{r\in \R^k} &c^\top \Phi r\nn\\
\text{s.t}\mb &\Phi r\geq T \Phi r.
\end{align}
Unless specified otherwise we use $\tr$ to denote the solution to the ALP and $\tj=\Phi \tr$ to denote the corresponding approximate value function.
The following is a preliminary error bound for the ALP from \cite{ALP}:
\begin{theorem}\label{restate}
Let $\mathbf{1}$, i.e., the vector with all-components equal to $1$, be in the span of the columns of $\Phi$ and $c$ be a probability distribution. Then, if $\tilde{J}_c=\Phi \tilde{r}_c$ is an optimal solution to the ALP in \eqref{alp}, then $||J^*-\tj||_{1,c}\leq \frac{2}{1-\alpha}\min_{r}||J^*-\Phi r||_\infty$, where $||x||_{1,c}=\sum_{i=1}^n c(i)|x(i)|$.
\end{theorem}
For a more detailed treatment of the ALP and sophisticated bounds the reader is referred to \cite{ALP}. Note that the ALP is a linear program in $k$ ($<<n$) variables as opposed to the LP in \eqref{mdplpshort} which has $n$ variables. Nevertheless, the ALP has $nd$ constraints (same as the LP) which is an issue when $n$ is large and calls for constraint approximation/reduction techniques.
\subsection{Related Work}
\textbf{Constraint sampling and The RLP:}
The most important work in the direction of constraint reduction is constraint sampling \cite{CS} wherein a reduced linear program (RLP) is solved instead of the ALP. While the objective of the RLP is same as that of the ALP, the RLP has only $m<<nd$ constraints. These $m$ constraints are \emph{sampled} from the original $nd$ constraints of the ALP according to a special sampling distribution $\psi_{u^*,V}$, where $u^*$ is the optimal policy and $V$ is a Lyapunov function (see \cite{CS} for a detailed presentation). If $\tilde{r}$ and $\tilde{r}_{RLP}$ are the solutions to the ALP and the RLP respectively form \cite{CS} we know that $||J^*-\Phi\tilde{r}_{RLP}||_{1,c}\leq ||J^*-\Phi\tilde{r}||_{1,c}+\epsilon ||J^*||_{1,c}$. A major gap in the theoretical analysis is that the error bounds are known for only a specific RLP formulated using idealized assumptions, i.e., knowledge of $u^*$.\\
\textbf{Other works:} Most works in literature make use of the underlying structure of the problem to cleverly reduce the number of constraints of the ALP. A good example is \cite{gkp}, wherein the structure in factored linear functions is exploited. The use of basis function also helps constraint reduction in \cite{Mor-Kum}. In \cite{ALP-Bor} the constraints are approximated indirectly by approximating the square of the Lagrange multipliers. \cite{petrik} reduces the transitional error ignoring the representational and sampling errors. Empirical successes include repeated application of constraint sampling to solve Tetris \cite{CST}. \\
\textbf{Open Questions:} The fact that RLP works well empirically goads us to build a more elaborate theory for constraint reduction. In particular, one would like to answer the following questions related to constraint reduction in ALP that have so far remained open. \\
$\bullet$ As a natural generalization of the RLP, what happens if we define a generalized reduced linear program (GRLP) whose constraints are positive linear combinations of the original constraints of the ALP?\\
$\bullet$ Unlike \cite{CS} which provides error bounds for a specific RLP formulated using an idealized sampling distribution is it possible to provide error bounds for any GRLP (and as result any RLP)?
In this paper, we address both of the questions above.
\section{Generalized Reduced Linear Program}
We define the generalized reduced linear program (GRLP) as below:
\begin{align}\label{grlp}
\min_{r\in \chi} &c^\top \Phi r,\nn\\
\text{s.t}\mb & W^\top \Phi r\geq W^\top T \Phi r, 
\end{align}
where $W \in \R_+^{nd\times m}$ is an $nd\times m$ matrix with all positive entries and $\chi \subset \R^k$ is any bounded set such that $\hj \in \chi$. Thus the $i^{th}$ ($1\leq i\leq m$) constraint of the GRLP is a positive linear combination of the original constraints of the ALP, see Assumption~\ref{superassump}. Constraint reduction is achieved by choosing $m<<nd$. Unless specified otherwise we use $\hr$ to denote the solution to the GRLP in \eqref{grlp} and $\hj=\Phi \hr$ to denote the corresponding approximate value function.
We assume the following throughout the rest of the paper:
\begin{assumption}\label{superassump}
$W \in \R_+^{nd\times m}$ is a full rank $nd\times m$ matrix with all non-negative entries. The first column of the feature matrix $\Phi$ (i.e.,$\phi_1$) is $\one$\footnote{$\one$ is a vector with all components equal to $1$. This definition is used throughout the paper.}  $\in \R^n$ and that $c=(c(i),i=1,\ldots,n)\in \R^n$ is a probability distribution, i.e., $c(i)\geq 0$ and $\sum_{i=1}^n c(i)=1$. It is straightforward to see that a RLP is trivially a GRLP.
\end{assumption}
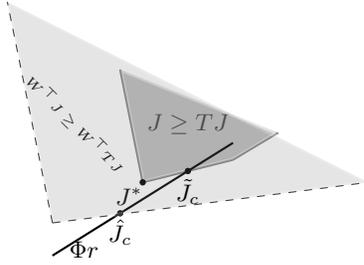
\begin{figure}[h!]
\centering
\begin{tikzpicture}[domain=-10:7.7,scale=0.6,font=\small,axis/.style={very thick, ->, >=stealth'}]
\draw[line,thick,-](0.5,3.5)--(1,1);
\draw[line,thick,-](1,1)--(3,1.5);
\draw[line,thick,-](3,1.5)--(4,2.1);
\node[](one) at (2,2.3){\text{$J\geq TJ$}};
\node[rotate=-45](seven) at (-0.5,2.3){\text{\tiny $W^\top J\geq W^\top TJ$}};
\node[](two) at (-0.3,-0.5){\text{$\Phi r$}};
\node[](three) at (0.7,0.7){\text{$J^*$}};
 \draw [ultra thick, draw=white, fill=gray, opacity=0.5]
       (0.5,3.5)--(1,1)--(3,1.5)--(4,2.1) -- cycle;
\draw[line,thick,-](-1,-0.625)--(3,1.8750);
 \fill (1,1)  circle[radius=2pt];
 \fill (2,1.25)  circle[radius=2pt];
 \fill (0.5,0.3125)  circle[radius=2pt];
\draw[line,dashed,-](-1,0.1)--(6,1);
\draw[line,dashed,-](-1,0.1)--(-2,5);
 \draw [ultra thick, draw=white, fill=gray, opacity=0.2]
       (-1,0.1)--(6,1)--(-2,5) -- cycle;
\node[] (four) at (2,0.8){\text{$\tilde{J}_c$}};
\node[] (six) at (0.5,-0.1){\text{$\hat{J}_c$}};
\end{tikzpicture}
\caption{The outer lightly shaded region corresponds to GRLP constraints and the inner dark shaded region corresponds to the original constraints. The main contribution of the paper is to provide a bound for $||J^*-\hat{J}_c||$.}
\label{cartoon}
\end{figure}
As a result of constraint reduction the feasible region of the GRLP is a superset of the feasible region of the ALP (see Figure~\ref{cartoon}). In order to bound $||J^*-\tj||$, \cite{ALP} makes use of the property that $\Phi \tr \geq T\Phi \tr$. However in the case of the GRLP this property does not hold anymore and hence it is a challenge to bound the error $||J^*-\hj||$. We tackle this challenge by introducing two novel $\max$-norm contraction operators called the least upper bound projection (LUBP) and approximate least upper bound projection operators (ALUBP) denoted by $\Gamma$ and $\tg$ respectively. We first present some definitions before the main result and a sketch of its proof.
The least upper bound (LUB) projection operator $\Gamma \colon \R^n \ra\R^n$ is defined as below:
\begin{definition}\label{lubpop}
Given $J\in \R^n$, its least upper bound projection is denoted by $\Gamma J$ and is defined as 
\begin{align}\label{gamdef}
(\Gamma J)(i)\stackrel{\Delta}{=}\underset{j=1,\ldots,k}{\min} (\Phi r_{e_j})(i), \mb \forall i=1,\ldots,n,
\end{align}
where $V(i)$ denotes the $i^{th}$ component of the vector $V\in \R^n$. Also in \eqref{gamdef}, $e_j$ is the vector with $1$ in the $j^{th}$ place and zeros elsewhere, and $r_{e_j}$ is the solution to the linear program in \eqref{lubplp} for $c=e_j$.
\begin{align}\label{lubplp}
r_c\stackrel{\Delta}{=}\min_{r\in \chi} &c^\top \Phi r,\nn\\
\text{s.t}\mb &\Phi r\geq  TJ.
\end{align}
\end{definition}
\begin{remark}
\mb\\
\begin{enumerate}
\item Observe that $\Gamma J\geq TJ$ (follows from the fact that if $a\geq c$ and $b\geq c$ then $\min(a,b)\geq c$, where $a, b, c \in \R$).
\item Given $\Phi$ and $J\in \R^n$, define $\F\stackrel{\Delta}{=}\{\Phi r|\Phi r\geq TJ\}$. Thus $\F$ is the set of all vectors in the span of $\Phi$ that upper bound $TJ$. By fixing $c$ in the linear program in \eqref{lubplp} we select a unique vector $\Phi r_c \in \F$. The LUB projection operator $\Gamma$ picks $n$ vectors $\Phi r_{e_i},i=1,\ldots,n$ from the set $\F$ and $\Gamma J$ is obtained by computing their component-wise minimum.
\item Even though $\Gamma J$ does not belong to the span of $\Phi$, $\Gamma J$ in some sense collates the various best upper bounds that can be obtained via the linear program in \eqref{lubplp}.
\item The LUB operator $\Gamma$ in \eqref{gamdef} bears close similarity to the ALP in \eqref{alp}.
\end{enumerate}
\end{remark}
We define an approximate least upper bound (ALUB) projection operator which has a structure similar to the GRLP and is an approximation to the LUB operator.
\begin{definition}\label{alubpop}
Given $J\in \R^n$, its approximate least upper bound (ALUB) projection is denoted by $\tg J$ and is defined as 
\begin{align}\label{tgamdef}
(\tg J)(i)\stackrel{\Delta}{=}\underset{j=1,\ldots,k}{\min} (\Phi r_{e_j})(i), \mb \forall i=1,\ldots,n,
\end{align}
where $r_{e_j}$ is the solution to the linear program in \eqref{alubplp} for $c=e_j$, and $e_j$ is same as in Definition~\ref{lubpop}.
\begin{align}\label{alubplp}
r_c\stackrel{\Delta}{=}\min_{r\in \chi} &c^\top \Phi r,\nn\\
\text{s.t}\mb &W^\top \Phi r\geq W^\top TJ, W \in \R^{nd\times m}_+ .
\end{align}
\end{definition}
\begin{definition}\label{bestproj}
The LUB projection of $J^*$ is denoted by $\bj=\Gamma J^*$, and let $r^*\stackrel{\Delta}{=}\underset{r\in \R^k}{\arg\min}||J^*-\Phi r^*||$.
\end{definition}
\subsection{Main Result}
\begin{theorem}\label{mr1}
\begin{align}\label{finalbnd}
||J^*-\hj||_{1,c}\leq\frac{6 ||J^*-\Phi r^*||_\infty+2||\Gamma\bj-\tg\bj||_\infty}{1-\alpha}.
\end{align}
\end{theorem}
\begin{proof}
Here we provide a sketch of the proof. Figure~\ref{schematic} gives an idea of the steps that lead to the result. First, one shows that the operators $\Gamma$ and $\tg$ have the $\max$-norm contraction property with factor $\alpha$. As a result, operators $\Gamma$ and $\tg$ have fixed points $\tv\in \R^n$ and $\hv \in \R^n$ respectively. This leads to the inequalities $\tj\geq \tv \geq J^*$ and $\hj \geq \hv$ (see Figure~\ref{schematic}), followed by which one can bound the term $||J^*-\hv||_{\infty}$ and then go on to show that any solution $\tr$ to the GRLP is also a solution to the program in \eqref{grlpeqprog}.
\begin{align}\label{grlpeqprog}
\min_{r\in \chi} &||\Phi r-\hv||_{1,c}\nn\\
\text{s.t}\mb & W^\top \Phi r\geq W^\top T \Phi r.
\end{align}
One then obtains the bound $||J^*-\hj||_{1,c}$ as in \eqref{finalbnd} using the fact that $||J^*-\bj||_{\infty}\leq 2||J^*-\Phi r^*||_{\infty}$ where $r^*$ is as in Definition~\ref{bestproj}.
\end{proof}\\
\FloatBarrier
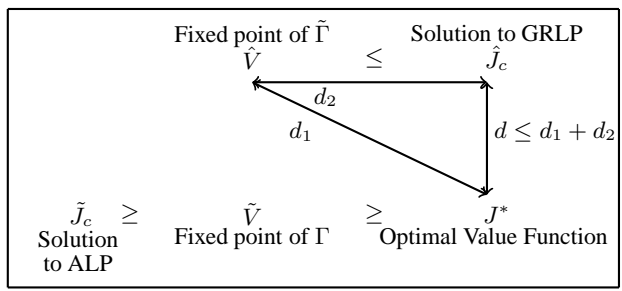
\begin{figure}[h!]
\centering
\begin{tikzpicture}[domain=0:7.7,scale=0.65,font=\small,axis/.style={very thick, ->, >=stealth'}]
\draw [line,thick,-] (0,-1)--(0,4.7);
\draw [line,thick,-] (0,4.7)--(12.5,4.7);
\draw [line,thick,-] (12.5,4.7)--(12.5,-1);
\draw [line,thick,-] (0,-1)--(12.5,-1);
\node[](one) at (1.5,0.5) {$\tj$};
\node[](four) at (2.5,0.5) {$\geq$};
\node[](two) at (5,0.5) {$\tv$};
\node[](three) at (10,0.5) {$J^*$};
\node[](five) at (7.5,0.5) {$\geq$};
\node[](six) at (10,3.7) {$\hj$};
\node[](seven) at (5,3.7) {$\hv$};
\node[](eight) at (7.5,3.7) {$\leq$};
\node[](nine) at(1.5,0){\text{Solution }};
\node[](twenty) at(1.5,-0.5	){\text{to ALP }};
\node[](ten) at(5,0){\text{Fixed point of }$\Gamma$};
\node[](eleven) at(10,0){\text{Optimal Value Function }};
\node[](twelve)at (10,4.2){\text{Solution to GRLP}};
\node[](thirteen)at (5,4.2){\text{Fixed point of }$\tg$};
\draw [line,thick,<->] (5,3.2)--(9.8,0.9);
\draw [line,thick,<->] (5,3.2)--(9.8,3.2);
\draw [line,thick,<->] (9.8,3.2)--(9.8,0.9);
\node[](fourteen)at (6,2.2){$d_1$};
\node[](fifteen)at (6.5,2.9){$d_2$};
\node[](sixteen)at (11.2,2.2){$d\leq d_1+d_2$};
\end{tikzpicture}
\caption{A schematic of the error analysis. \quad\quad\quad\quad\quad Here  $d=||J^*-\hj||_{1,c}.$}
\label{schematic}
\end{figure}
\vspace{-10pt}
It is important to note that computing $\Gamma/\tg$ involves solving $n$ linear programs which is easy when $n$ is small, however, the same becomes difficult and impractical when $n$ is large. Nevertheless, we hasten to point out that these quantities are only analytical constructs that lead us to the error bounds, and need not be calculated in practice for systems with large $n$.
\subsection{Result Discussion}
We now make various important qualitative observations about the result in Theorem~\ref{mr1}.\\
\textbf{Error Terms:}
The error term is split into two factors, the first of which is related to the best possible projection while the second factor is related to constraint approximation. The second factor $||\Gamma \bj-\tg\bj||_\infty$ is completely defined in terms of $\Phi$, $W$ and $T$, and does not require knowledge of stationary distribution of the optimal policy. It makes intuitive sense since given that $\Phi$ approximates $J^*$, it is enough for $W$ to depend on $\Phi$ and $T$ without any additional requirements. Unlike the result in \cite{CS} which holds only for a specific RLP formulated under ideal assumptions, our bounds hold for any GRLP and as a result for any given RLP. Another interesting feature of our result is that it holds with probability $1$. Also by making use of appropriate Lyapunov functions as in \cite{ALP}, the error bound in \eqref{finalbnd} can also be stated using a weighted $L_\infty$-norm, thereby indicating the relative importance of states.\\
\textbf{Additional insights on constraint sampling:}
It is easy to notice from Definitions~\ref{lubpop}, ~\ref{alubpop} and \ref{bestproj} that for any given state $s\in S$, $\Gamma \bj(s)\geq J^*(s)$, and that $\Gamma\bj(s)\geq\tg \bj(s)$. If the state $s$ is selected in the RLP, then it is also true that $\Gamma \bj(s)\geq \tg \bj(s)\geq J^*(s)$. Thus the additional error $|\Gamma \bj(s) -\tg \bj(s)|$ due to constraint sampling is less than the original projection error $|\Gamma \bj(s)-J^*(s)|$ due to function approximation. This means that the RLP is expected to perform well whenever \emph{important} states are retained after constraint sampling. Thus the sampling distribution need not be the stationary distribution of the optimal policy as long as it samples the important states, an observation that might theoretically explain the empirical successes of the RLP \cite{ALP,CST,SALP}.\\
\textbf{Relation to other ADP methods:}
\FloatBarrier
\begin{table}[H]
\resizebox{\columnwidth}{!}{
\begin{tabular}{|c|c|c|}\hline
ADP Method	&Empirical	&Theoretical\\\hline
Projected Bellman &\ding{51}    &\ding{53}-Policy Chattering\\
Equation	&\cite{lstd,lspe,lspi}	&\cite{BertB}\\\hline  
ALP		&\ding{53}-Large Constraints &\ding{51}-\cite{ALP}\\\hline		
RLP		&\ding{51} -\cite{CST,ALP,SALP} &\ding{53}-  Only under\\
&&ideal assumptions\\\hline
\end{tabular}
}
\end{table}
A host of the ADP methods such as \cite{lspi,lspe,lstd,Tsit} are based on solving the projected Bellman equation (PBE). The PBE based methods have been empirically successful and also have theoretical guarantees for the approximate value function. However, a significant shortcoming is that they suffer from issue of \emph{policy-chattering} (see section~$6.4.3$ of \cite{BertB}), i.e., the sequence of policies might oscillate within a set of bad policies. A salient feature of the ALP based methods is that they find only one approximate value function $\tj$ and one sub-optimal policy derived as a greedy policy with respect to $\tj$. As a result there is no such issue of policy-chattering for the ALP based methods. By providing the error bounds for the GRLP, our paper provides the much required theoretical support for the RLP. Our GRLP framework closes the long-standing gap in the literature of providing a theoretical framework to bound the error due to constraint reduction in ALP based schemes.\\
\textbf{GRLP is linear function approximation of the constraints:}
In order to appreciate this fact consider the Lagrangian of the ALP and GRLP in \eqref{lag} and \eqref{lag2} respectively, i.e., 
\begin{align}\label{lag}
\tilde{L}(r,\lambda)=c^\top \Phi r+\lambda^\top (T\Phi r-\Phi r), \\ \label{lag2}\hat{L}(r,q)=c^\top \Phi r+q^\top W^\top (T\Phi r-\Phi r).
\end{align}
The insight that the GRLP is linear function approximation of constraints (i.e., the Lagrangian multipliers) can be obtained by noting that $ Wq\approx \lambda$ in \eqref{lag2}. Note that while the ALP employs LFA in its objective, the GRLP employs linear approximation both in the objective as well as the constraints. This has significance in the context of the reinforcement learning setting \cite{Sutton} wherein the model information is available in the form of noisy sample trajectories. RL algorithms make use of stochastic approximation (SA) \cite{SA} and build on ADP methods to come up with incremental update schemes to learn from noisy samples presented to them. An SA scheme to solve the GRLP in RL setting can be derived in a manner similar to \cite{ALP-Bor}.
\section{Application to Controlled Queues}\label{exmp}
We take up an example in the domain of controlled queues to show that experiments confirm with the theory developed. More specifically, we look at the error bounds for different constraints reduction schemes to demonstrate the fact that whenever value of $\et$ is less the GRLP solution is closer the optimal value function.\\
The queuing system consists of $n=10^4$ states and $d=4$ actions. We chose $n=10^4$ because it was possible to solve both the GRLP and the exact LP (albeit with significant effort) so as to enumerate the approximation errors. We hasten to mention that while we could run the GRLP for queuing systems with $n>10^4$ without much computational overhead, solving the exact LP was not possible for $n>10^4$ as a result of which the approximation error could not be computed.\\
\textbf{Queuing Model:}
The queuing model used here is similar to the one in Section~$5.2$ of \cite{ALP}. We consider a single queue with arrivals and departures. The state of the system is the queue length with the state space given by $S=\{0,\ldots,n-1\}$, where $n-1$ is the buffer size of the queue. The action set $A=\{1,\ldots,d\}$ is related to the service rates. We let $s_t$ denote the state at time $t$. The state at time $t+1$ when action $a_t \in A $ is chosen is given by $s_{t+1}= s_{t}+1$ with probability $p$, $s_{t+1}= s_{t}-1$ with probability $q(a_t)$ and $s_{t+1}= s_t$, with probability $(1-p-q(a_t)$. For states $s_t=0$ and $s_t=n-1$, the system dynamics is given by 	$s_{t+1}= s_{t}+1$ with probability $p$ when $s_t=0$ and $s_{t+1}=s_t-1$ with probability $q(a_t)$ when $s_t=n-1$.
The service rates satisfy $0<q(1)\leq \ldots\leq q(d)<1$ with $q(d)>p$ so as to ensure `stabilizability' of the queue. The reward associated with the action $a \in A$ in state $s\in S$ is given by $g_a(s)=-(s+60q(a)^3)$.\\
\textbf{Choice of $\Phi:$} We make use of polynomial features in $\Phi$ (i.e., $1,s,\ldots,s^{k-1}$) since they are known to work well for this domain \cite{ALP}. This takes care of the term $||J^*-\Phi r^*||_\infty$ in \eqref{finalbnd}. \\
\textbf{Selection of $W$:} For our experiments, we choose two contenders for the $W$-matrix:\\
{$\mathbf{(i)}$} $W_c$- matrix that corresponds to sampling according to $c$. This is justified by the insights obtained from the error term $\et$ and the idea of selecting the important states.\\
{$\mathbf{(ii)}$} $W_a$ state-aggregation matrix, a heuristic derived by interpreting $W$ to be the feature matrix that approximates the Lagrange multipliers as $\lambda\approx Wq$, where $\lambda \in \R^{nd}, r\in \R^m$. One can show \cite{dolgov} that the optimal Lagrange multipliers are the discounted number of visits to the ``state-action pairs'' under the optimal policy $u^*$, i.e., 
\begin{align}
\lambda^*(s,u^*(s))&=\big(c^\top(I-\alpha P_{u^*})^{-1}\big)(s)\nn\\
				&= \big(c^\top(I+\alpha P_{u^*}+\alpha^2 P_{u^*}^2+\ldots)\big)(s).\nn\\
			\lambda^*(s,u^*(s))&=0, \forall a \neq u^*(s).\nn
\end{align}
where $P_{u^*}$ is the probability transition matrix with respect to the optimal policy. Even though we might not have the optimal policy in practice $u^*$, the fact that $\lambda^*$ is a linear combination of $\{P_{u^*},P^2_{u^*},\ldots\}$ hints at the kind of features that might be useful for the $W$ matrix. Our choice of $W_a$ matrix to correspond to aggregation of near by states is motivated by the observation that $P^n$ captures $n^{th}$ hop connectivity/neighborhood information.
The aggregation matrix $W_a$ is defined as below: $\forall i=1,\ldots,m$,
\begin{align}\label{wdes}
W_a(i,j)&=1, \mb\forall j\mb\text{s.t}\mb j=(i-1)\times\frac{n}{m}+k+(l-1)\times n, \nn\\&\mb\quad\quad k=1,\ldots,\frac{n}{m}, l=1,\ldots,d,\nn\\
&=0,\mb\text{otherwise}.
\end{align}
In order to provide a contrast between good and bad choices of $W$ matrices we also make use of two more matrices, an ideal matrix $W_i$ generated by sampling according to the stationary distribution of the optimal policy as in \cite{CS} and $W_c$ generated by sampling using $c$ and  $W_r$ a random matrix in $\R^{nd\times m}_+$. For the sake of comparison we compute $||\Gamma\bj-\tg\bj||_\infty$ for the different $W$ matrices.
 Though computing $||\Gamma\bj-\tg\bj||_\infty$ might be hard in the case of large $n$, since $||\Gamma\bj-\tg\bj||_\infty$ is completely dependent on the structure of $\Phi$, $T$ and $W$ we can compute it for small $n$ instead and use it as a surrogate. Accordingly, we first chose a smaller system $Q_S$ with $n=10$, $d=2$, $k=2$, $m=5$, $q(1)=0.2$, $q(2)=0.4$, $p=0.2$ and $\alpha=0.98$. In the case of $Q_S$, $W_a$ (\eqref{wdes} with $m=5$) turns out to be a $20 \times 5$ matrix where the $i^{th}$ constraint of the GRLP is the average of all constraints corresponding to states $(2i -1)$ and $2i$ (there are four constraints corresponding to these two states). The various error terms are listed in Table~\ref{errterms} and plots are shown in Figure~\ref{q1}. It is clear from Table~\ref{errterms} that $W_a$, $W_i$ and $W_c$ have much better $\et$ than randomly generated positive matrices. Since each constraint is a hyperplane, taking linear combinations of non-adjacent hyperplanes might drastically affect the final solution. This could be a reason why $W_r$ (random matrix) performs badly in comparison with other $W$ matrices.
\FloatBarrier
\begin{table}[H]
\centering
\begin{tabular}{|c|c|c|c|c|}\hline
Error Term&	$W_i$&	$W_c$& $W_a$& $W_r$ \\\hline
$\et$ & $39$	&$84$& $54.15$& $251.83$ \\\hline
\end{tabular}
\caption{Shows various error terms for $Q_S$.}
\label{errterms}
\end{table}
Having validated the choices of $W$s on $Q_S$ we then consider a moderately larger queuing system (denoted by) $Q_L$ with $n=10^4$ and $d=4$ with $q(1)=0.2$, $q(2)=0.4$, $q(3)=0.6$, $q(4)=0.8$, $p=0.4$ and $\alpha=0.98$. In the case of $Q_L$ we chose $k=4$ (i.e., we used $1, s,s^2$ and $s^3$ as basis vectors) and we chose $W_a$ \eqref{wdes}, $W_c$, $W_i$ and $W_r$ with $m=50$. We set $c(s)=(1-\zeta) \zeta^s, \mb\forall s=1,\ldots,9999$, with $\zeta=0.9$ and $\zeta=0.999$ respectively. The results in Table~\ref{pref} show that performance exhibited by $W_a$ and $W_c$ are better by several orders of magnitude over `random' in the case of the large system $Q_L$ and is closer to the ideal sampler $W_i$. Also note that a better performance of $W_a$ and $W_c$ in the larger system $Q_L$ tallies with a lower value of $\et$ in the smaller system $Q_S$.
\FloatBarrier
\begin{table}[H]
\resizebox{\columnwidth}{!}{
\begin{tabular}{|c|c|c|c|c|}\hline
Error Terms&	$W_i$&	$W_c$& $W_a$& $W_r$ \\\hline
$||J^*-\hj||_{1,c}$ for $\zeta=0.9$& $32$&	$32$& $220$& $5.04\times 10^4$ \\\hline
$||J^*-\hj||_{1,c}$ for $\zeta=0.999$& $110$&	$180.5608$& $82$& $1.25\times 10^7$ \\\hline
\end{tabular}
}
\caption{Shows performance metrics for $Q_L$.}
\label{pref}
\end{table}
\section{Conclusion}
Solving MDPs with large number of states is of practical interest. However, when the number of states is large, it is difficult to calculate the exact value function. ALP is a widely studied ADP scheme that computes an approximate value function and offers theoretical guarantees. Nevertheless, the ALP is difficult to solve due to its large number of constraints and in practice a reduced linear program (RLP) is solved. Though RLP has been shown to perform well empirically, theoretical guarantees are available only for a specific RLP formulated under idealized assumptions. This paper provided a more elaborate treatment of constraint reduction/approximation. Specifically, we generalized the RLP to formulate a generalized reduced linear program (GRLP) and provided error bounds. Our results solved a major open problem of analytically justifying linear function approximation of the constraints. We discussed the implications of our results in the contexts of ADP and reinforcement learning. We demonstrated the fact that experiments conform to the theory developed in this paper via an example in the domain of controlled queues. Future directions include providing more sophisticated error bounds based on Lyapunov functions, a two-time scale actor-critic scheme to solve the GRLP, and basis function adaptation schemes to tune the $W$ matrix.
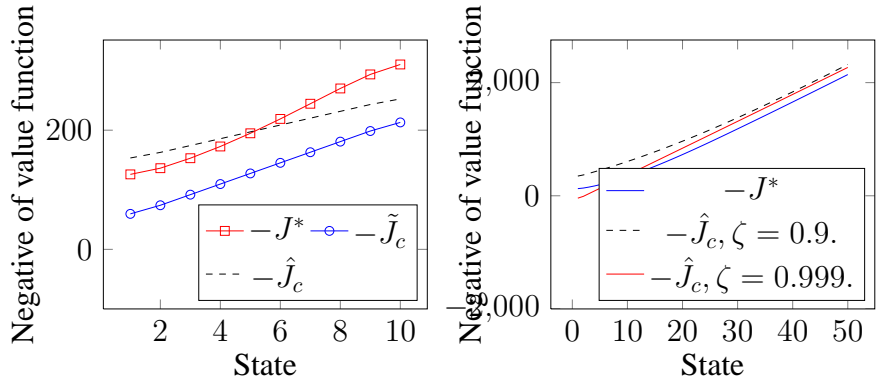
\begin{figure}[H]
\resizebox{\columnwidth}{!}{
\begin{tabular}{cc}
\begin{tikzpicture}
\begin{axis}[legend pos=south east,legend columns=2, scale=0.75,ymin=-100,xlabel=\Large State,ylabel=\Large Negative of value function]
\addplot[mark=square,red]  plot file {V};
\addplot[mark=o,blue]  plot file {V_alp};
\addplot[dashed,black]  plot file {V_rlp};
\Large
\addlegendentry{$-J^*$}
\addlegendentry{$-\tj$}
\addlegendentry{$-\hj$}
\end{axis}
\end{tikzpicture}
\begin{tikzpicture}
\begin{axis}[legend pos=south east,legend columns=1, scale=0.75,ymin=-2000,xlabel=\Large State,ylabel=\Large{N}egative of value function]
\addplot[mark=.,black,blue]  plot file {Vl};
\addplot[mark=.,dashed,black]  plot file {Vl_rlp};
\addplot[mark=.,red]  plot file {Vl999_rlp};
\Large
\addlegendentry{$-J^*$}
\addlegendentry{$-\hj, \zeta=0.9.$}
\addlegendentry{$-\hj, \zeta=0.999.$}
\end{axis}
\end{tikzpicture}
\end{tabular}
}
\caption{Plot corresponding to $Q_S$ on the left and $Q_L$ on the right. The GRLP here used $W_a$ in \eqref{wdes} with $m=5$ for $Q_S$ and $m=50$ for $Q_L$.}
\label{q1}
\end{figure}
\bibliography{ref.bib}
\bibliographystyle{plain}
\newpage
\onecolumn
\begin{appendix}
\section{Proofs}
We present the proofs for the Lemmas and Theorems stated in the main body of the paper. As and when required we also state and prove other intermediate lemmas. For the sake of clarity we restate Assumption~\ref{superassump} as Assumption~\ref{wassump}-~\ref{probdist} below:\\
\begin{assumption}\label{wassump}
$W \in \R^{nd\times m}_+$ is a full rank $nd\times m$  matrix (where $m<<nd$) with all non-negative entries, and $\Phi$ is an $n\times k$ feature matrix (where $k<<n$).
\end{assumption}
\begin{assumption}\label{one}
The first column of the feature matrix $\Phi$ (i.e.,$\phi_1$) is $\one \in \R^n$. In other words, the constant function is part of the basis.
\end{assumption}
\begin{assumption}\label{probdist}
$c=(c(i),i=1,\ldots,n)\in \R^n$ is a probability distribution, i.e., $c(i)\geq 0$ and $\sum_{i=1}^n c(i)=1$.
\end{assumption}
We now state without proof the following properties of $T$:
\begin{lemma}\label{monotone}
\textbf{Monotonicity:} Let $J_1,J_2\in \R^n$ be such that $J_1\geq J_2$, then $TJ_1\geq TJ_2$.
\end{lemma}
\begin{lemma}\label{shift}
\textbf{Shifting:} For any $J \in \R^n$ and $\one \in R^n$ be a vector with all components $1$\footnote{This definition of $\one$ is the same throughout the paper.} and $k \in \R$ be a constant, then $T(J+k\one)=TJ+\alpha k \one$.
\end{lemma}
\begin{lemma}\label{contra}
\textbf{Contraction:} For any $J_1,J_2\in \R^n$, $||TJ_1-TJ_2||_\infty\leq\alpha ||J_1-J_2||_\infty$.
\end{lemma}
Lemmas~\ref{monotone},~\ref{shift} and ~\ref{contra} are standard in MDP literature and can be found in \cite{BertB}.  

\begin{lemma}\label{cseq}
The RLP in $(2)$ of \cite{CS} obtained via sampling the constraints is a special case of GRLP.
\end{lemma}
\begin{proof}
Let the constraints of the ALP be numbered from $1\to nd$ and $q_1,\ldots,q_m$ denote the $m$ sampled constraints. The GRLP with $W$ defined as 
\begin{align}
W(i,j)&=1, \mb\text{if}\mb q_i=j\nn\\
&=0, \mb\text{otherwise}
\end{align}
is the RLP with the corresponding sampled constraints.
\end{proof}
\begin{lemma}\label{subset}
Let $r_f\in \R^k$ be any feasible solution to the ALP in \eqref{alp}, then it is also feasible for the GRLP in \eqref{grlp}.
\end{lemma}
\begin{proof}
Follows from the fact that $W$ has all positive entries and that each constraint of the GRLP is a positive linear combination of original constraints in the ALP.
\end{proof}
\begin{lemma}\label{bestbnd}
Let $r^*\in \R^k$ be defined as $r^*\stackrel{\Delta}{=}\arg\min_{r\in R^k}||J^*-\Phi r||_\infty$, then 
\begin{align}
||J^*-\bj||_\infty\leq 2||J^*-\Phi r^*||_\infty.
\end{align}
\end{lemma}
\begin{proof}
The result follows from the definition of $\Gamma$ in \eqref{gamdef}, Assumption~\ref{one} and the fact that $\Phi r^*+||J^*-\Phi r^*||_\infty \one\geq TJ^*$.
\end{proof}
\begin{lemma}\label{gmonotone}
For $J_1, J_2\in \R^n$ such that $J_1\geq J_2$, we have $\Gamma J_1\geq \Gamma J_2$.
\end{lemma}
\begin{proof}
Choose any $i\in \{1,\ldots,n\}$ and let $r^1_{e_i}$ and $r^2_{e_i}$ be the unique solutions to the linear program in \eqref{lubplp} for $c=e_i$ with $J=J_1$ and $J=J_2$ respectively. Since $J_1\geq J_2$, we have $TJ_1\geq TJ_2$ and $e_i^\top \Phi r^1_{e_i} \geq e_i^\top \Phi r^2_{e_i}$, i.e., $(\Phi r^1_{e_i})(i)\geq (\Phi r^2_{e_i})(i)$. The proof follows from the fact that $(\Gamma J)(i)=(\Phi r_{e_i})(i),\mb\forall J\in \R^n$, and our choice of $i$ was arbitrary.
\end{proof}\\\vspace{10pt}
\begin{lemma}\label{gshift}
Let $J_1\in \R^n$ and $k\in \R$ be a constant. If $J_2=J_1+k\one$, then $\Gamma J_2=\Gamma J_1+\alpha k\one$.
\end{lemma}
\begin{proof}
Choose any $i\in \{1,\ldots,n\}$, let $r^1_{e_i}$ and $r^2_{e_i}$ be the unique solutions to linear program in \eqref{lubplp} for $c=e_i$ with $J=J_1$ and $J=J_2$ respectively. By Assumption~\ref{one} and Lemma~\ref{shift}, we know that $r^1_{e_i}+\alpha k e_1$ is feasible for the $i^{th}$ linear program associated with $\Gamma J_2$ and we claim that $r^2_{e_i}=r^1_{e_i}+\alpha ke_1$. On the contrary, if $r^2_{e_i}\ne r^1_{e_i}+\alpha ke_1$, then $(\Phi r^2_{e_i})(i)< (\Phi r^1_{e_i}+\alpha k e_1)(i)$ (since the solution to the linear program in \eqref{lubplp} is unique) and since $r^2_{e_i}-\alpha k e_1$ is feasible for the $i^{th}$ linear program associated with $\Gamma J_1$ we will have $(\Phi r^2_{e_i}-k\alpha e_1)(i)<(\Phi r^1_{e_i})(i)$. Thus we have arrived at a contradiction because we assumed that $r^1_{e_i}$ is the unique solution for the $i^{th}$ linear program associated with $\Gamma J_1$. So 
\begin{align}\label{equality}
r^2_{e_i}=r^1_{e_i}+\alpha k e_1, \mb\forall i \in \{1,\ldots,n\},\mb \text{since $i$ was arbitrary}.
\end{align}
From \eqref{equality} and Assumption~\ref{one} it follows that $\Gamma J_2=\Gamma J_1+\alpha k \one$.\\\vspace{10pt}
\end{proof}
\begin{theorem}\label{gmaxcontra}
The operator $\Gamma  \colon \R^n\ra \R^n$ obeys the $\max$-norm contraction property with factor $\alpha$.
\end{theorem}
\begin{proof}
Given $J_1,J_2\in \R^n$ let $\epsilon=||J_1-J_2||_\infty$. Thus
\begin{align}\label{ineq}
J_2-\epsilon\one\leq J_1\leq J_2+\epsilon \one.
\end{align}
From Lemmas~\ref{gmonotone} and ~\ref{gshift} we can write
\begin{align}\label{ineq}
\Gamma J_2-\alpha \epsilon\one\leq \Gamma J_1\leq \Gamma J_2+\alpha \epsilon\one.
\end{align}
\end{proof}
One can show that the following iterative scheme in \eqref{pvi} based on the LUB projection operator $\Gamma$ in \eqref{gamdef} converges to a unique fixed point $\tv$.
\begin{align}\label{pvi}
V_{n+1}&=\Gamma V_n,\mb \forall n\geq 0.
\end{align}
\begin{lemma}\label{gfp}
 $\tv$, the unique fixed point of the iterative scheme \eqref{pvi}, obeys $\tv\geq T\tv$.
\end{lemma}
\begin{proof}
Consider the $i^{th}$ linear program associated with $\Gamma \tv$. We know that $\Phi r_{e_i}\geq T \tv,\mb \forall i=1\to n$. The result follows from noting that $\tv$ is an unique fixed point of $\Gamma $ and that $\tv(i)=\underset{j=1\to n}{\min}(\Phi r_{e_j})(i)$.
\end{proof}
\begin{lemma}\label{relation1}
 $\tv$, the unique fixed point of the iterative scheme \eqref{pvi}, and the solution $\tj$ to the ALP in \eqref{alp}, obey the relation $\tj\geq\tv\geq J^*$.
\end{lemma}
\begin{proof}
Since $\tv\geq T\tv$ it follows that $\tv\geq J^*$. Let $\Phi r_1, \Phi r_2,\ldots,\Phi r_n$ be solutions to the ALP in \eqref{alp} for $c=e_1, e_2,\ldots,e_n$ respectively. Now consider the iterative scheme in \eqref{pvi} with $V_0(i)=\underset{j=1\to n}{\min}(\Phi r_j)(i)$. It is clear from the definition of $V_0$ that $\tj\geq V_0$. Also from monotone property of $T$ we have $\Phi r_i\geq T \Phi r_i\geq T V_0, \mb\forall i=1\to n$ and hence $V_0\geq TV_0$. Since $V_1=\Gamma V_0$, from the definition of $\Gamma$ in \eqref{gamdef} we have $V_0\geq V_1$, and recursively $V_{n}\geq V_{n+1}, \mb\forall n\geq 0$. So it follows that $\tj\geq V_0\geq V_1\ldots\geq \tv$.
\end{proof}
\begin{theorem}\label{fxpres}
Let $\tv$ be the fixed point of the iterative scheme in \eqref{pvi} and let $\bj$ be the best possible projection of $J^*$ as in Definition~\ref{bestproj}, then
\begin{align}
||J^*-\tv||_\infty\leq \frac{1}{1-\alpha}||J^*-\bj||_\infty.
\end{align}
\end{theorem}
\begin{proof}
Let $\epsilon=||J^*-\bj||_\infty$, and $\{V_n\},n\geq 0$ be the iterates of the scheme in \eqref{pvi} with $V_0=\bj$, then
\begin{align}
||J^*-\tv||_\infty&\leq ||J^*-V_0+V_0-V_1+V_1\ldots-\tv||_\infty\nn\\\vspace{10pt}
&\leq ||J^*-V_0||_\infty+||V_0-V_1||_\infty+||V_1-V_2||_\infty+\ldots\nn\\\vspace{10pt}
&\quad\quad(\text{Since $||V_1-V_0||_\infty=||\Gamma \bj-\Gamma J^*||_\infty\leq\alpha||\bj-J^*||_\infty$, from Theorem~\ref{gmaxcontra}})\nn\\\vspace{10pt}
&\leq \epsilon+\alpha\epsilon+\alpha^2\epsilon+\ldots\nn\\\vspace{10pt}
&=\frac{\epsilon}{1-\alpha}.
\end{align}
\end{proof}
\begin{lemma}\label{tgmonotone}
For $J_1, J_2\in \R^n$ such that $J_1\geq J_2$, we have $\tg J_1\geq \tg J_2$.
\end{lemma}
\begin{proof}
Proof follows from Assumptions~\ref{wassump} and ~\ref{one} using arguments along the lines of Lemma~\ref{gmonotone}.
\end{proof}
\begin{lemma}\label{tgshift}
Let $J_1\in \R^n$ and $k\in \R$ be a constant. If $J_2=J_1+k\one$, then $\tg J_2=\tg J_1+\alpha k\one$.
\end{lemma}
\begin{proof}
Proof follows from Assumption~\ref{wassump} and~\ref{one} using arguments along the lines of Lemma~\ref{gshift}.
\end{proof}
\begin{theorem}\label{tgmaxcontra}
The operator $\tg \colon \R^n\ra \R^n$ obeys the $\max$-norm contraction property with factor $\alpha$ and the following iterative scheme based on the ALUB projection operator $\tg$, see \eqref{apvi}, converges to a unique fixed point $\hv$.
\begin{align}\label{apvi}
V_{n+1}&=\tg V_n,\mb\forall n\geq 0.
\end{align}
\end{theorem}
\begin{proof}
Follows on similar lines of proof of Theorem~\ref{gmaxcontra}.
\end{proof}
\begin{lemma}\label{relation2}
The unique fixed point $\hv$ of the iteration in \eqref{apvi} and the solution $\hj$ of the GRLP obey $\hj\geq\hv$.
\end{lemma}
\begin{proof}
Follows in a similar manner as the proof for Lemma~\ref{relation1}.
\end{proof}
\begin{theorem}\label{mt1}
Let $\hv$ be the fixed point of the iterative scheme in \eqref{apvi} and let $\bj$ be the best possible approximation of $J^*$ as in Definition~\ref{bestproj}, then
\begin{align}
||J^*-\hv||_\infty\leq \frac{||J^*-\bj||_\infty+||\Gamma \bj-\tg\bj||_\infty}{1-\alpha}.
\end{align}
\end{theorem}
\begin{proof}
Let $\epsilon=||J^*-\bj||_\infty$, and $\{V_n\},n\geq 0$ be the iterates of the scheme in \eqref{apvi} with $V_0=\bj$, then
\begin{align}
||\bj-\tg\bj||_\infty&\leq||\bj-\Gamma \bj||_\infty+||\Gamma\bj-\tg\bj||_\infty\nn\\\vspace{10pt}
&=||\Gamma J^*-\Gamma \bj||_\infty+||\Gamma\bj-\tg\bj||_\infty\nn\\\vspace{10pt}
&\leq \alpha\epsilon+\beta,
\end{align}
where $\beta=||\Gamma \bj-\tg\bj||_\infty$. Now
\begin{align}
||J^*-\hv||_\infty&\leq ||J^*-V_0+V_0-V_1+V_1\ldots-\hv||_\infty\nn\\\vspace{10pt}
&\leq ||J^*-V_0||_\infty+||V_0-V_1||_\infty+||V_1-V_2||_\infty+\ldots\nn\\\vspace{10pt}
&=||J^*-V_0||_\infty+||V_0-V_1||_\infty+||\tg V_0-\tg V_1||_\infty+\ldots\nn\\\vspace{10pt}
&\leq \epsilon+(\beta+\alpha\epsilon)+\alpha(\beta+\alpha\epsilon)+\ldots\nn\\\vspace{10pt}
&=\frac{\epsilon+\beta}{1-\alpha}.
\end{align}
\end{proof}
\begin{theorem}\label{cmt1}
Let $\hv$, $\bj$ be as in Theorem~\ref{mt1} and let $r^*\stackrel{\Delta}{=}\arg\min_{r\in \R^k}||J^*-\Phi r||_\infty$ then
\begin{align}
||J^*-\hv||_\infty\leq \frac{2||J^*-\Phi r^*||_\infty+||\Gamma \bj-\tg\bj||_\infty}{1-\alpha}.
\end{align}
\end{theorem}
\begin{proof}
The result is obtained by using Lemma~\ref{bestbnd} to replace the term $||J^*-\bj||_\infty$ in Theorem~\ref{mt1}.
\end{proof}
\begin{lemma}\label{srw}
$\hat{r} \in \R^k$ is a solution to GRLP in \eqref{grlp} iff it solves the following program:
\begin{align}\label{grlpeqprog}
\min_{r\in \chi} &||\Phi r-\hv||_{1,c}\nn\\
\text{s.t}\mb & W^\top \Phi r\geq W^\top T \Phi r.
\end{align}
\end{lemma}
\begin{proof}
We know from Lemma~\ref{relation2} that $\hj\geq\hv$, and thus minimizing $||\Phi r-\hv||_{1,c}=\sum_{i=1}^n c(i) |(\Phi r)(i)-\hv(i)|=c^\top \Phi r-c^\top \hv$, is same as minimizing $c^\top \Phi r$.
\end{proof}
\begin{theorem}\label{mt2}
Let $\hv$ be the solution to the iterative scheme in \eqref{apvi} and let $\hj=\Phi \hr$ be the solution to the GRLP. Let $\bj$ be the best possible approximation to $J^*$ as in Definition~\ref{bestproj}, and $||\Gamma \bj -\tg\bj||_\infty$ be the error due to ALUB projection and let $r^*\stackrel{\Delta}{=}\arg\min||J^*-\Phi r||_\infty$, then
\begin{align}
||\hj-\hv||_{1,c}\leq\frac{4||J^*-\Phi r^*||_\infty+||\Gamma\bj-\tg\bj||_\infty}{1-\alpha}.
\end{align}
\end{theorem}
\begin{proof}
Let $\gamma=||J^*-\Phi r^*||_\infty$, then  it is easy to see that
\begin{align}
||J^*-T\Phi r^*||_\infty&=||TJ^*-T\Phi r^*||_\infty\leq\alpha\gamma,\mb\text{and}\nn\\\vspace{10pt}
||T\Phi r^*-\Phi r^*||_\infty&\leq(1+\alpha)\gamma.
\end{align}
From Assumption~\ref{one} there exists $r'\in \R^k$ such that $\Phi r'=\Phi r^*+\frac{(1+\alpha)\gamma}{1-\alpha}\one$ and $r'$ is feasible to the ALP. Now
\begin{align}
||\Phi r'-J^*||_\infty\leq ||\Phi r^* -J^*||_\infty+||\Phi r'-\Phi r^*||_\infty\leq \gamma+\frac{(1+\alpha)\gamma}{1-\alpha}=\frac{2\gamma}{1-\alpha}.
\end{align}
Since $r'$ is also feasible for GRLP in \eqref{grlp} we have
\begin{align}
||\hj-\hv||_{1,c}&\leq||\Phi r'-\hv||_{1,c}\nn\\\vspace{10pt}
&\leq||\Phi r'-\hv||_\infty\mb\text{(Since $c$ is a distribution)}\nn\\\vspace{10pt}
&\leq||\Phi r'-J^*||_\infty+||J^*-\hv||_\infty \mb \quad\quad\text{(From Corollary~\ref{cmt1} we have)}\nn\\\vspace{10pt}
&\leq\frac{4\gamma+\beta}{1-\alpha}
\end{align}
\end{proof}
\begin{corollary}\label{cmt2}
Let $\hj$, $\hv$, $r^*$ and $J^*$ be as in Theorem~\ref{mt2}, then
\begin{align}\label{finalbnd}
||J^*-\hj||_{1,c}\leq\frac{6 ||J^*-\Phi r^*||_\infty+2||\Gamma\bj-\tg\bj||_\infty}{1-\alpha}.
\end{align}
\end{corollary}
\begin{proof}
\begin{align}
||J^*-\hj||_{1,c}&\leq||J^*-\hv||_{1,c}+||\hv-\hj||_{1,c}\nn\\\vspace{10pt}
&\leq||J^*-\hv||_\infty+||\hv-\hj||_{1,c}\nn
\end{align}
\end{proof}	
The result is obtained by using Corollary~\ref{cmt1} for the first term and Theorem~\ref{mt2} for the second term in the above inequality.\\\vspace{10pt}
\section{Numerical Example of Single Queue with Finite Buffer size and Controlled Service Rates}
The problem setting we consider is similar to the one presented in Sections~$5.2$ and $6.1$ in \cite{ALP}. However, we provide the most important details in this section so as to make the material self contained.\\
We consider a single queue with finite buffer size where the maximum allowed queue length is $n-1$. The queue evolves in discrete instants of time $t=0,1,\ldots$ with only one of the following mutually exclusive events occurring between $t$ and $t+1$
\begin{itemize}
\item A job arrives with probability $p$.
\item A job gets served and leaves the queue with probability $q(a)$. Here $a \in A=\{1,\ldots,d\}$ is an action.
\end{itemize}
It is understood that excess jobs (i.e., jobs arriving when the queue length is $n-1$) will be discarded.\\
Formally, the dynamics of the controlled queuing system can be described via the framework of Markov Decision Process (MDP). The state space is given by $S=\{0,1,\ldots,n-1\}$ and denotes the number of jobs waiting in the queue. The action set is given by $A=\{1,\ldots,d\}$ and controls the probability of a job getting serviced and leaving the queue. We let $s_t$ denote the state at time $t$. At time $t$, for $0<s_t<n-1$, the state at time $t+1$ when action $a_t \in A $ is chosen is given by
\begin{align}\label{sysdef}
	s_{t+1}	&= s_{t}+1, \mb\text{with probability}\mb p,\nn\\\vspace{10pt}
		&= s_{t}-1, \mb\text{with probability}\mb q(a_t) ,\nn\\\vspace{10pt}
		&= s_t, \mb\text{with probability}\mb (1-p-q(a_t)).
\end{align}
For states $s_t=0$ and $s_t=n-1$ the system dynamics is given by
\begin{align}\label{sysdef1}
	s_{t+1}	&= s_{t}+1, \mb\text{with probability}\mb p,\mb\text{when}\mb s_t=0\nn\\\vspace{10pt}
		&= s_{t}-1, \mb\text{with probability}\mb q_{a_t},\mb\text{when}\mb s_t=n-1.
\end{align}
In order to ensure `stabilizability', we assume the following condition on the system:
\begin{align}
0<q(1)\leq \ldots\leq q(d)<1,\mb\text{ where}\mb q(d)>p.
\end{align}
Note that the above state transition description in \eqref{sysdef} and \eqref{sysdef1} has been presented in a concise format in Section~\ref{exmp}.
 The reward associated with the action $a \in A$ in state $s\in S$ is given by
\begin{align}\label{rew}
g_a(s)=-(s+60q(a)^3).
\end{align}
The reward function is negative in queue length since it is desirable to penalize higher queue length. One can also observe that \eqref{rew} penalizes actions that offer higher level of service.
\section{Solution via GRLP }\label{ALPmethods}
We present the solution methodology in Figure~\ref{sbsp}.
\begin{figure}[h!]
\begin{tikzpicture}[domain=-10:7.7,scale=0.7,font=\small,axis/.style={very thick, ->, >=stealth'}]
\draw[line,thick,-] (-21.5,-0.5)--(-21.5,0.8);
\draw[line,thick,-] (-21.5,0.8)--(-3.5,0.8);
\draw[line,thick,-] (-3.5,0.8)--(-3.5,-0.5);
\draw[line,thick,-] (-3.5,-0.5)--(-21.5,-0.5);
\node[](one)at (-20,0.5){\text{Choose right $c$}};
\node[](two)at (-20,0){\text{\& $\Phi$}};
\draw [line,thick,->] (-18.5,0.5)--(-18,0.5);
\node[](one)at (-16,0.5){\text{Obtain the right ALP}};
\draw [line,thick,->] (-14,0.5)--(-13.5,0.5);
\node[](one)at (-10.2,0.5){\text{Choose a \emph{good} $W$ by computing}};
\node[](one)at (-10.2,0.0){\text{$||\Gamma\bj-\tg\bj||_\infty$ for various $W$s}};
\draw [line,thick,->] (-7.0,0.5)--(-6.5,0.5);
\node[](one)at (-5,0.5){\text{Arrive at the}};
\node[](one)at (-5,0){\text{ right GRLP}};
\end{tikzpicture}
\caption{A step by step method to arrive at the right GRLP.}
\label{sbsp}
\end{figure}
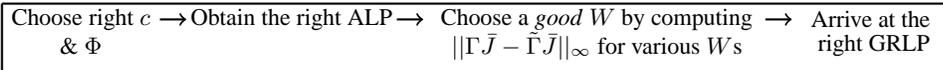
\vspace{10pt}
\subsection{Choice of $\Phi$ and $c$}
The polynomial features are known to work well for this problem. \cite{ALP} gives a proper justification of this choice using arguments based on Lyapunov function.\\
A good choice of $c$ is 
\begin{align}\label{cchoice}
c(s)=(1-\zeta) \zeta^s, \mb\forall s=0,1,\ldots,n-1.
\end{align}
$c$ is the state relevance weight vector and denotes the relative importance of the various states. By choosing $c$ as in \eqref{cchoice}, one can give importance to smaller queue lengths compared to large queue lengths. This choice is also supported by the fact that the stationary probability $\pi(s)$ of the state $s$ in the case of a stable uncontrolled queue is of the form $\pi(s)\propto (\frac{\rho}{1-\rho})^s$. However, when $n$ is small (say $10$) $c$ can have a uniform distribution, since the ratio of $\frac{c(n)}{c(0)}=\zeta^n$ is close to $1$ for $\zeta$ sufficiently close to $1$ (say $\zeta=0.99$), i.e., the state relevance weights do not decay much for small $n$.
\subsection{Choice of $W$}
Grouping adjacent states as in \eqref{wdes} (presented in the main body of the paper) might be a good idea. This choice of $W$ was validated by Table~\ref{errterms} and we provide further insights by presenting the active and passive constraints of both ALP and GRLP for smaller system $Q_S$ in Figure~\ref{feasible}.
\FloatBarrier
\begin{figure}[H]
\begin{tabular}{cc}
\begin{tikzpicture}
\begin{axis}[scale=0.8, transpose legend,
legend columns=1,
legend style={at={(0.5,-0.1)},anchor=north},
]
\addplot[domain=-60:60,thick,smooth,variable=\y,black]  plot ({-74.000000+-1.000000*\y},{\y});
\addplot[domain=-60:60,dashed,variable=\y,black]  plot ({-124.000000+-2.000000*\y},{\y});
\addplot[mark=*,black] plot file{r_ALP};
\addplot[domain=-60:60,thick,smooth,variable=\y,black]  plot ({-242.000000+-10.800000*\y},{\y});
\addplot[domain=-60:60,dashed,variable=\y,black]  plot ({-292.000000+-11.800000*\y},{\y});
\addplot[domain=-60:60,dashed,variable=\y,black]  plot ({-24.000000+9.800000*\y},{\y});
\addplot[domain=-60:60,dashed,variable=\y,black]  plot ({-192.000000+9.800000*\y},{\y});
\addplot[domain=-60:60,dashed,variable=\y,black]  plot ({-174.000000+-3.000000*\y},{\y});
\addplot[domain=-60:60,dashed,variable=\y,black]  plot ({-342.000000+-12.800000*\y},{\y});
\addplot[domain=-60:60,dashed,variable=\y,black]  plot ({-224.000000+-4.000000*\y},{\y});
\addplot[domain=-60:60,dashed,variable=\y,black]  plot ({-392.000000+-13.800000*\y},{\y});
\addplot[domain=-60:60,dashed,variable=\y,black]  plot ({-274.000000+-5.000000*\y},{\y});
\addplot[domain=-60:60,dashed,variable=\y,black]  plot ({-442.000000+-14.800000*\y},{\y});
\addplot[domain=-60:60,dashed,variable=\y,black]  plot ({-324.000000+-6.000000*\y},{\y});
\addplot[domain=-60:60,dashed,variable=\y,black]  plot ({-492.000000+-15.800000*\y},{\y});
\addplot[domain=-60:60,dashed,variable=\y,black]  plot ({-374.000000+-7.000000*\y},{\y});
\addplot[domain=-60:60,dashed,variable=\y,black]  plot ({-542.000000+-16.800000*\y},{\y});
\addplot[domain=-60:60,dashed,variable=\y,black]  plot ({-424.000000+-8.000000*\y},{\y});
\addplot[domain=-60:60,dashed,variable=\y,black]  plot ({-592.000000+-17.800000*\y},{\y});
\addplot[domain=-60:60,dashed,variable=\y,black]  plot ({-474.000000+-18.800000*\y},{\y});
\addplot[domain=-60:60,dashed,variable=\y,black]  plot ({-642.000000+-28.600000*\y},{\y});
\addlegendentry{Active Constraints}
\addlegendentry{Inactive Constraints}
\addlegendentry{$\tr$}
\end{axis}
\end{tikzpicture}
&
\begin{tikzpicture}
\begin{axis}[scale=0.8,xmin=-500,xmax=400,transpose legend,
legend columns=1,
legend style={at={(0.5,-0.1)},anchor=north},
]
\addplot[domain=-60:60,smooth,variable=\y,black]  plot ({-133.000000+1.950000*\y},{\y});
\addplot[domain=-60:60,dashed,variable=\y,red]  plot ({-74.000000+-1.000000*\y},{\y});
\addplot[domain=-60:60,dotted,variable=\y,black]  plot ({-333.000000+-9.400000*\y},{\y});
\addplot[mark=*,green] plot file{r_RLP};
\addplot[mark=o,blue] plot file{r_ALP};
\addplot[domain=-60:60,smooth,variable=\y,black]  plot ({-233.000000+-7.400000*\y},{\y});
\addplot[domain=-60:60,dotted,variable=\y,black]  plot ({-433.000000+-11.400000*\y},{\y});
\addplot[domain=-60:60,dotted,variable=\y,black]  plot ({-533.000000+-18.300000*\y},{\y});
\addplot[domain=-60:60,dashed,variable=\y,red]  plot ({-242.000000+-10.800000*\y},{\y});
\addlegendentry{Active Constraints of GRLP}
\addlegendentry{Active Constraints of ALP}
\addlegendentry{Inactive Constraints of GRLP}
\addlegendentry{$\tj$}
\addlegendentry{$\hj$}
\end{axis}
\end{tikzpicture}
\end{tabular}
\caption{Constraints of the ALP (left) and the GRLP for system $Q_S$  with $n=10$, $d=2$, $k=2$, $m=5$, $q(1)=0.2, q(2)=0.4, p=0.2$ and $\alpha=0.98$. In this case $c$ has a uniform distribution.
The dotted and solid lines in the right plot show the inactive and active constraints of the GRLP respectively, the dashed lines in the right plot show the active constraints of the ALP. The feasible region in both cases (ALP \& GRLP) are to the right of the corresponding active constraints.}
\label{feasible}
\end{figure}
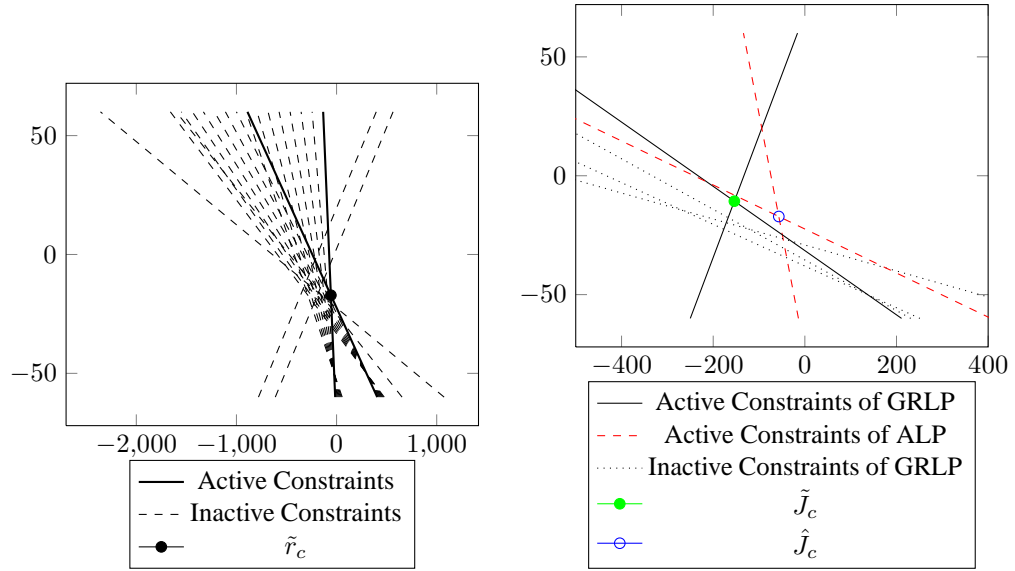
The following plot shows the functions related to $\et$ (for system $Q_S$):
\FloatBarrier
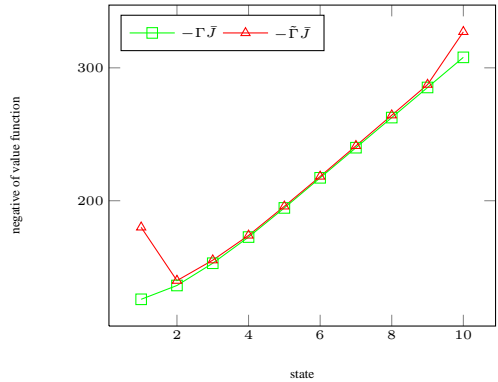
\begin{figure}[H]
\centering
\begin{tikzpicture}
\begin{axis}[legend pos=north west, legend columns=2, scale=0.75, xlabel=state, ylabel=negative of value function]
\addplot[mark=square,green]  plot file {./gjb};
\addplot[mark=triangle,red]  plot file {./gtjb};
\tiny
\addlegendentry{$-\Gamma\bj$}
\addlegendentry{$-\tg\bj$}
\end{axis}
\end{tikzpicture}
\end{figure}
The following plot shows the various error terms (for system $Q_S$):
\FloatBarrier
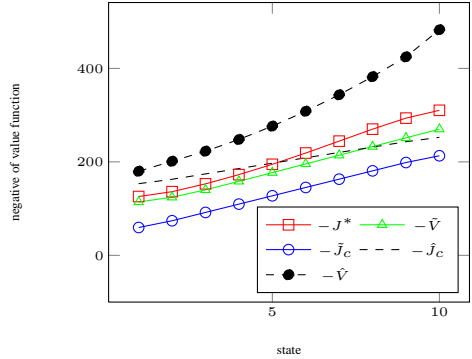
\centering
\begin{tikzpicture}
\begin{axis}[legend pos=south east,legend columns=2, scale=0.7,ymin=-100,xlabel=state,ylabel=negative of value function]
\addplot[mark=square,red]  plot file {./V};
\addplot[mark=triangle,green]  plot file {./V_g};
\addplot[mark=o,blue]  plot file {./V_alp};
\addplot[dashed,black]  plot file {./V_rlp};
\addplot[dashed,mark=*,black]  plot file {./V_t};
\tiny
\addlegendentry{$-J^*$}
\addlegendentry{$-\tv$}
\addlegendentry{$-\tj$}
\addlegendentry{$-\hj$}
\addlegendentry{$-\hv$}
\end{axis}
\end{tikzpicture}
\subsection{Performance of Greedy Policy}
We define the greedy policy $\hu$ as the one which is greedy with respect to $\hj$, i.e.,
\begin{align}\label{greedypol}
\hu(s)&=\arg\max_{ a\in A}\big(g_a(s)+\alpha \sum_{s'}p_a(s,s')\hj(s')\big).
\end{align}
The following plot shows the performance of $\hu$ in the case of the two systems ($Q_S$ and $Q_L$).
\FloatBarrier
\begin{figure}[H]
\begin{tabular}{cc}
\begin{tikzpicture}
\begin{axis}[legend pos=south east,legend columns=1, scale=0.7,ymin=-200,xlabel=state, ylabel=negative of value function]
\addplot[mark=.,black]  plot file {./V};
\addplot[mark=.,dashed]  plot file {./V_pol_rlp};
\tiny
\addlegendentry{$-J^*$}
\addlegendentry{$-J_{\hu}$}
\end{axis}
\end{tikzpicture}
&
\begin{tikzpicture}
\begin{axis}[legend pos=south east,legend columns=1, scale=0.7,ymin=-200,xlabel=state, ylabel=negative of value function]
\addplot[mark=.,black]  plot file {./Vl};
\addplot[mark=.,dashed]  plot file {./Vl_pol_rlp};
\addplot[mark=.,dotted]  plot file {./Vl999_pol_rlp};
\tiny
\addlegendentry{$-J^*$}
\addlegendentry{$-J_{\hu}, \zeta=0.9.$}
\addlegendentry{$-J_{\hu}, \zeta=0.999.$}
\end{axis}
\end{tikzpicture}
\end{tabular}
\caption{Plot corresponding to $Q_S$ on the left and $Q_L$ on the right.}
\label{polplot}
\end{figure}
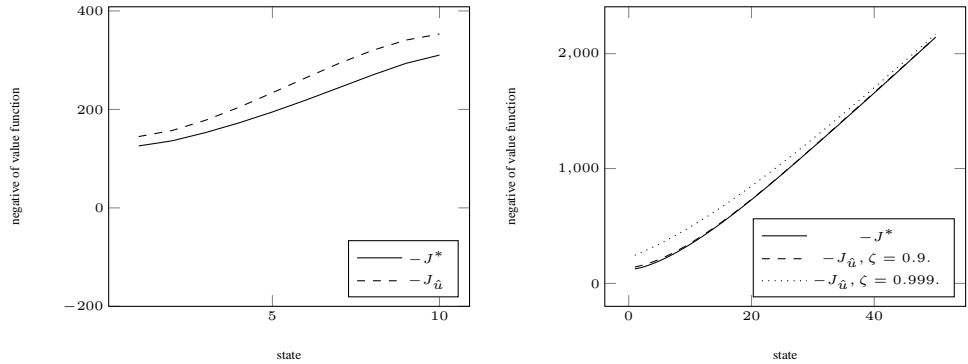
It can be seen from the figures that $J_{\hu}$ is close to $J^*$. In particular, in the case of $Q_L$, $J_{\hu}$ is nearly the same as $J^*$ for $\zeta=0.999$.
\end{appendix}
\end{document}